\documentclass[a4paper,UKenglish,cleveref, autoref, thm-restate]{lipics-v2021}

\titlerunning{}

\author{Michał Fica}{University of Wroc\l{}aw}{}{}{} \author{Jan Otop}{University of Wroc\l{}aw}{}{}{}

\authorrunning{Michał Fica and Jan Otop}

\Copyright{Michał Fica and Jan Otop}

\title{Active Automata Learning with Advice}


\category{} 

\relatedversion{}

\nolinenumbers 

\hideLIPIcs

\usepackage{latexsym}
\usepackage{amssymb}
\usepackage{amsmath}
\usepackage{amsthm}
\usepackage{booktabs}
\usepackage{enumitem}
\usepackage{graphicx}
\usepackage{color}

\usepackage{soul}
\usepackage{url}
\usepackage[utf8]{inputenc}
\usepackage{amsmath}
\usepackage{amsthm}
\usepackage{booktabs}
\usepackage{algorithm}
\usepackage{algpseudocode}

\usepackage{amssymb}
\usepackage{tikz}
\usepackage{xspace}
\usepackage{amsthm}
\usepackage{thmtools,thm-restate}
\usepackage{tcolorbox}

\usepackage{todonotes}
\presetkeys{todonotes}{inline,backgroundcolor=green}{}

\usetikzlibrary{arrows,automata,shapes,positioning}
\usetikzlibrary{matrix,decorations.pathmorphing,decorations.pathreplacing,calligraphy}

\newcommand{\set}[1]{\{#1\}}
\newcommand{\aut}{\mathcal{A}}
\newcommand{\autB}{\mathcal{B}}

\newcommand{\PSPACE}{\textsc{PSpace}{}}

\newcommand{\lang}{\mathcal{L}}

\newcommand{\Paragraph}[1]{\noindent\textbf{#1.}}

\newcommand{\M}{\mathcal{M}}
\newcommand{\rightCongL}{\sim_{L}}
\newcommand{\rightCongF}{\sim_{f}}

\newcommand{\rewritesOneStep}{\rightarrow}
\newcommand{\rewrites}{\rewritesOneStep^*}
\newcommand{\normalForm}[1]{\textbf{NF}(#1)}
\newcommand{\normalFormSRS}[2]{\textbf{NF}_{#1}(#2)}

\newcommand{\rew}{\mathcal{R}}
\newcommand{\lStar}{L\ensuremath{^*}}

\newcommand{\bigO}{\mathcal{O}}

\newcommand{\Iff}{\leftrightarrow}

\newcommand{\learnlib}{\textbf{\texttt{learnlib}}}

\newcommand{\Ind}{\textsc{Ind}}
\newcommand{\rewConv}{\rew_{\textsc{Con}}}

\newcommand{\BibTeX}{B\kern-.05em{\sc i\kern-.025em b}\kern-.08em\TeX}

\funding{
This work was supported by the National Science Centre (NCN), Poland under grant 2024/53/B/ST6/01620.}

\begin{document}

\maketitle

\begin{abstract}
We present an extended automata learning framework that combines active automata learning with deductive inference.
The learning algorithm asks membership and equivalence queries as in the original framework, but it is also given advice, 
which is used to infer answers to queries when possible and reduce the burden on the teacher. 
We consider advice given via string rewriting systems, which specify equivalence of words w.r.t. the target languages.
The main motivation for the proposed framework is to reduce the number of queries. We show how to adapt Angluin-style learning algorithms to this framework with low overhead. 
Finally, we present empirical evaluation of our approach and observe substantial improvement in query complexity.
\end{abstract}

\section{Introduction}
\label{s:intro}
\Paragraph{Active automata learning}
In the \emph{active automata learning} framework the learning algorithm 
generates a \emph{deterministic finite automaton (DFA)} recognizing an unknown regular language 
$\lang$ (target language) based on the following queries about the language: 
\emph{membership} of a word in $\lang$, and 
\emph{equivalence} of a language of a DFA $\aut$ with $\lang$, where the answer to the latter is either YES, 
or a \emph{counterexample} word distinguishing $\lang$ and $\lang(\aut)$.
This framework has been proposed in the seminal paper~\cite{angluin1987learning} along
 with the \lStar-algorithm implementing learning in polynomial time. 
Since then, active automata learning has been intensively researched 
for both efficiency~\cite{TTTalgorithm,ADTalgorithm,Lsharp} as well as extensions to other automata-based models 
such as Mealy machines~\cite{DBLP:conf/iccad/PenaO98,DBLP:conf/fm/ShahbazG09},
weighted automata~\cite{activeWeightedAutomata} and other models~\cite{DBLP:conf/dlt/DrewesH03,MarusicW15,MoermanS0KS17}.

\Paragraph{Automata learning in applications} 
Active automata learning is a practical method for {synthesis} of 
abstract models from complex systems, called \emph{model learning}.
Model learning has been applied to a wide-rage of systems including Smartcards, ESM Controllers, or 
network protocols such as TCP and SSH~\cite{tcp-learning,ssh-learning,vaandrager2017model}.
In such applications, membership queries are answered automatically by running the existing system. 
However, answering equivalence queries is difficult as equivalence of a program with an automaton is
an undecidable problem and hence it has been only approximated in these applications~\cite{tcp-learning,ssh-learning}.

\Paragraph{Query complexity} 
The number of queries asked by the \lStar-algorithm is polynomial in the size of 
the constructed automaton, but answering even a single query may be expensive.
The main objective of this work is to reduce the number of queries asked by the active learning algorithm 
(called query complexity).
There is a large body of work on improving query complexity of 
active automata learning~\cite{TTTalgorithm,ADTalgorithm,Lsharp,KrugerJR24,DierlFHJST24}.
However, possible gains in this line of work are limited due to minimal amount 
of information necessary to identify one of exponentially
many minimal DFA. Therefore, we consider a different approach, in which 
the learning algorithm is given additional information (called advice) to restrict the search space.

\Paragraph{Learning with advice}
We introduce an extended learning framework, in which the learning algorithm is 
given \emph{advice} about the target language. While the algorithm 
still actively asks membership and equivalence queries, the advice constraints the search space and allows to infer answers to some queries.
Therefore, it combines two approaches to synthesis: deductive synthesis from the specification and
the inductive synthesis based on queries. 
The rationale behind advice is to specify clearly true properties of the model under learning such as certain actions
(e.g. closing a connection) being idempotent.  
In this way, the learning algorithm considers only automata consistent with the advice as candidate automata, 
which reduces query complexity. 
Furthermore, it makes the framework more flexible as some aspects of the constructed automaton can be specified 
with examples (queries) while other with advice.
We consider advice given via \emph{string rewriting systems} and their extension \emph{controlled string rewriting systems}.

\Paragraph{String rewriting systems as advice}
A string rewriting system (SRS) is a set of rewrite rules $l \rightarrow r$, which can be regarded as directed 
identities $l = r$. The transformation of one word into another by replacing an infix $l$ with $r$ iteratively is called \emph{rewriting}.
In our approach, an advice SRS specifies equivalence w.r.t. the language, i.e., 
if $w$ can be rewritten to $v$, then either both belong to the language or none of them belongs to the language. 
Advice is by design not complete; not all equivalent words can be rewritten and their equivalence 
is to be learned with queries. 
The SRS formalism is conceptually different than the automata model, and hence
some properties, which are difficult to express with DFA or examples can be easily expressed with SRS. 
We discuss different types of advice expressed with SRS in Section~\ref{s:power}.

\Paragraph{Overview}
In Section~\ref{s:framework}, we introduce the framework for active automata learning with SRS advice and 
discuss algorithms that infer answers to membership and equivalence queries using SRS advice.
Then, in Section~\ref{s:power}, we discuss properties that can be expressed with advice given via SRS.
We demonstrate viability of the proposed approach with empirical evaluation of our algorithm, where 
we present improvement in query complexity (Section~\ref{s:experiments}).
Finally, we discuss other approaches to define advice in Section~\ref{s:extensions}: controlled string rewriting systems and one-sided advice.

\Paragraph{Related work}
There is a large body of work on optimization of computational and query complexity of 
the \lStar-algorithm~\cite{TTTalgorithm,ADTalgorithm,Lsharp,DBLP:conf/fm/ShahbazG09,KrugerJR24,DierlFHJST24}. 
These results are mostly independent to our work as our algorithm can be adapted to 
any implementation of the \lStar-algorithm. 

Specializing the \lStar-algorithm to improve its performance on particular subclasses of automata is an active research topic.
In the \emph{gray-box learning} framework~\cite{graybox-learning,DBLP:conf/dlt/BerthonBPR21}, 
the algorithm learns the sequential composition of Mealy machines $\mathcal{M}$ and $\mathcal{N}$ knowing 
$\mathcal{M}$ in advance. 
The goal of that work is to learn $\mathcal{M}$ more efficiently than learning the whole sequential composition.
Another example are Mealy machines that are parallel compositions of other Mealy machines~\cite{DBLP:conf/fossacs/LabbafGHM23,DBLP:journals/corr/abs-2405-08647,ProductAutomata}. 
The notion of the parallel composition is the counterpart of the \emph{convolution}, which we consider in this paper. 
We did not compare the results from~\cite{DBLP:conf/fossacs/LabbafGHM23,DBLP:journals/corr/abs-2405-08647}
 with our approach as there are important differences. First, in our framework the decomposition needs to be given in order to define the SRS, while
in~\cite{DBLP:conf/fossacs/LabbafGHM23,DBLP:journals/corr/abs-2405-08647} the decomposition is learned. 
Second, we present a general framework, while the algorithms in~\cite{DBLP:conf/fossacs/LabbafGHM23,DBLP:journals/corr/abs-2405-08647} has been specialized for parallel compositions.

This work is conceptually related to the SyGuS framework~\cite{sygus}, 
which combines inductive synthesis and constraints on the synthesized program.
While inductive synthesis is built on the active automata learning framework, it is not restricted to automata. 
Another difference is the type of constraints; while we consider semantic advice, which refers to the target language, 
constraints in SyGuS are syntactic and given via a context-free grammar. 

There has been work on improving learning by integrating human into a learning loop (\emph{machine coaching}) to
give a piece of advice explaining wrong decisions made by ML models~\cite{DBLP:conf/comma/MarkosTM22}.
While machine coaching advice is interactive, advice in our approach is given before the learning process starts.

\section{Preliminaries}
\label{s:preliminaries}
Given a finite alphabet $\Sigma$ of letters, a \emph{word} $w$ is a finite sequence 
of letters. 
We denote the set of all finite words over $\Sigma$ by $\Sigma^*$.
For a word $w$, we define $w[i]$ as the $i$-th letter of $w$.

\Paragraph{DFA}
A \emph{deterministic finite automaton} (DFA) is a tuple $(\Sigma, Q, q_0, F, \delta)$
consisting of
    the alphabet $\Sigma$, 
    a finite set of states $Q$, 
    the initial state $q_0 \in Q$,  
    a set of final states $F$, and
    a transition function   $\delta \colon Q \times \Sigma \to Q$.
The transition function $\delta$ extends from letters to words in the canonical way;
we define $\widehat{\delta} \colon Q \times \Sigma^* \to Q$ for all $q \in Q$, $a \in \Sigma$ and
$w \in \Sigma^*$ as: $\widehat{\delta}(q,\epsilon) = q$ and 
$\widehat{\delta}(q,aw) = \widehat{\delta}\left(\delta(q,a),w\right)$. 
We abuse the notation by writing $\delta$ instead of $\widehat{\delta}$ to simplify notation.

\Paragraph{Semantics of DFA}
A \emph{run} $\pi$  of a DFA $\aut$ on a word $w$ is a sequence of states $\pi[0] \pi[1] \dots \pi[|w|]$ such that $\pi[0]$ is the initial state 
and for every $0 <i \leq |w|$ we have $\delta(\pi[i-1],w[i]) = \pi[i]$. A DFA has exactly one run on every word.
A run $\pi$ over $w$ is accepting if the last state belongs to the set of accepting states, i.e., $\pi(|w|) \in F$; in such a case we say that $w$ is accepted by $\aut$.
The language \emph{recognized} by a DFA $\aut$, denoted by $\lang(\aut)$, is the set of words accepted by $\aut$.

\Paragraph{Mealy machines and their semantics}
A Mealy machine is an extension of DFA, in which (a)~all transitions are labeled with output letters 
and (b)~all states are accepting.  
Formally, a Mealy machine $\M$ is a tuple $(\Sigma, \Gamma, Q, q_0, \delta, \lambda)$ such that $(\Sigma, Q, q_0, Q, \delta)$ is a DFA
and $\lambda \colon Q \times \Sigma \to \Gamma$ is a labeling function.
The output of $\M$ is a word over $\Gamma$ defined for every $w \in \Sigma^*$ 
as $\lambda(\pi[0],w[1]) \ldots \lambda(\pi[k-1],w[k])$, where $k = |w|$ and
$\pi$ is the run of $\M$ on $w$.
Then, the Mealy machine $\M$ defines a function $f_{\M} \colon  \Sigma^* \to \Gamma^*$ such that for every $w \in \Sigma^*$,
the word $f_{\M}(w) \in \Gamma^*$ is the output of $\M$ on $w$.

\subsection{Active automata learning} 
\label{s:preliminaries-lstar}
We briefly recall active automata learning and the \lStar-algorithm for regular languages~\cite{angluin1987learning}, which is based on the Myhill–Nerode theorem.

\Paragraph{Myhill–Nerode theorem}
For a language $\lang$ over $\Sigma$, we define the right congruence relation of $\lang$, 
denoted by $\rightCongL \subseteq \Sigma^* \times \Sigma^*$, as follows:
for all $u,v \in \Sigma^*$ we have $u \rightCongL v$ if and only if $\forall w. uw \in \lang \Iff vw \in \lang$. 
The Myhill–Nerode theorem states that $\lang$ is regular if and only if $\lang$ has finitely many equivalence classes.
Furthermore, it defines the minimal automaton $\aut_{\lang}$ recognizing $\lang$, i.e., the states of 
$\aut_{\lang}$ correspond to equivalence classes of $\rightCongL$, 
while the transition function and accepting states are uniquely defined.

\Paragraph{Queries} 
To learn the language $\lang$, the learning algorithm asks the following queries to an \emph{oracle for} $\lang$:
\begin{itemize}[noitemsep,topsep=0pt]
\item 
(MQ)~\emph{membership queries}: given $w \in \Sigma^*$, is $w \in \lang$?, and
\item 
(EQ)~\emph{equivalence queries}: given a DFA $\aut$, is $\lang(\aut) = \lang$? 
If not return a counterexample, which is a word from exactly one of the sets $\lang(\aut)$ and $\lang$.
\end{itemize}

\Paragraph{Overview of the algorithm} The \lStar-algorithm iteratively constructs the right congruence relation of 
the target language $\lang$, which defines the minimal DFA recognizing $\lang$.
It maintains two sets of words: the set of \emph{selectors} $S$ containing $\epsilon$, 
which correspond to states of the constructed automaton, and the set of \emph{test words} $C$, 
which defines the following approximation of the right congruence: 
$u \rightCongL^C v$ if and only if $\forall w \in C. uw \in \lang \Iff vw \in \lang$. 
The set $S$ contains exactly one selector of each \emph{reachable} equivalence class of $\rightCongL^C$.
The algorithm starts with $S = C = \set{\epsilon}$ and executes the following loop:
\begin{enumerate}[noitemsep,topsep=2pt]
\item Build a DFA $\aut_{S,C}$ with states $S$ and $\rightCongL^C$ defining the transition relation;
 $\rightCongL^C$ is computed with  membership queries.
\item Perform the equivalence query with $\aut_{S,C}$, and terminate if it is positive, i.e., $\lang = \lang(\aut_{S,C})$.
\item If the equivalence query returns a counterexample, process it to extend $S,C$, and repeat.
\end{enumerate}

The \lStar-algorithm computes the minimal DFA recognizing $\aut_{\lang}$ in polynomial time in $|\aut_{\lang}|$ and the total length of returned counterexamples.

\subsection{Rewriting}
\label{s:rewriting}
We introduce the necessary notions on string rewriting.
For details and examples see~\cite{book1993string}.

\Paragraph{String rewriting} 
A \emph{string rewriting system (SRS)} $\rew$ over an alphabet $\Sigma$ is 
a finite set of pairs of words $(l,r)$ over $\Sigma$.
A pair of words  $(l,r)$  from $\rew$ is called a \emph{string rewrite rule} and denoted with $l \rightarrow r$.
For an SRS $\rew$, we define a \emph{single-step rewrite relation} $\rewritesOneStep_{\rew}$ over words from $\Sigma^*$ as follows:
for all $s,t \in \Sigma^*$, we have $s \rewritesOneStep_{\rew} t$ if and only if there are words $x,y \in \Sigma^*$ and a rewrite rule
$(l,r) \in \rew$ such that $s = x l y$ and $t = x r y$. 
The (string) rewriting relation $\rewrites_{\rew}$ is the transitive and reflexive closure of $\rewritesOneStep_{\rew}$.
We will omit the $\rew$ subscript if the SRS is clear from the context.

\Paragraph{Normal forms}  
A word $t$ is in a \emph{normal form} if there is no $t'$ such that $t \rewritesOneStep_{\rew} t'$. 
If for a word $s$ there is a unique $t$ in a normal form such that $s \rewrites_{\rew} t$, we say
that $t$ is the normal form of $s$ (w.r.t. $\rew$) and denote it by $\normalFormSRS{\rew}{s}$.
We skip $\rew$ and simply write $\normalForm{s}$ if the SRS $\rew$ is clear from the context.

\Paragraph{Computing normal forms}  
A (finite) SRS $\rew$ is
\emph{terminating} if every sequence of words $s_0, s_1, \ldots$ 
    such that $s_i \rewritesOneStep s_{i+1}$ is finite, 
\emph{confluent} if for all words $s,s_1, s_2$, if $s \rewrites s_1$ and $s \rewrites s_2$, then there is $t$ such that
    $s_1 \rewrites t$ and $s_2 \rewrites t$, and
\emph{convergent} if it is terminating and confluent.
In a convergent SRS, every word $s$ has the (unique) normal form, 
which can be computed by applying reductions in arbitrary way as long as the word can be reduced. 
Termination guarantees that this process terminates and confluence entails the uniqueness of the result.
Therefore, for a convergent SRS $\rew$, one can effectively compute $\normalFormSRS{\rew}{s}$, and hence 
the following problem is decidable:
given words $s,t$, is there a word $u$ such that $s \rewrites u$ and $t \rewrites u$.
It suffices to compute normal forms $\normalForm{s},\normalForm{t}$ and check whether $\normalForm{s} = \normalForm{t}$.
While the above problem is decidable for convergent SRS, the property of being convergent is undecidable over SRS. 
More precisely, both termination and confluency are in general undecidable, 
but there is a large body of work on these properties~\cite{book1993string}.

\section{Learning automata with advice}
\label{s:framework}
\newcommand{\oracle}{\mathcal{O}}
\newcommand{\rewReset}{\rew_{\textsc{syn}}}
\newcommand{\rewIdm}{\rew_{\textsc{idm}}}

We present the framework for active learning automata with advice given via a string rewriting system. 
While we present the framework and the algorithms for DFA, we conclude this section 
with the discussion on extensions to other automata models.

\subsection{Learning framework}
The advice given to the learning algorithm specifies which words are equivalent w.r.t. the target language, 
i.e., the SRS rewrites words from the target language to other words in the language and, 
similarly, words not in the language to words not belonging to the language. 
The advice is intended to specify ground truth about the target language, while it is 
not expected to be complete; it does not need to relate all equivalent words.
However, wrong advice can prevent the algorithm from learning the target language.

We discuss a relaxed variant of advice (one-sided advice), which is related to the implication 
rather than equivalence, in Section~\ref{s:one-sided}.

\begin{definition}
\label{def:consistent}
An SRS $\rew$ is \emph{consistent} with a (regular) language $\lang$ if and only if
    for all words $s,t$,  if $s \rewrites_{\rew} t$ , then $s \in \lang \Iff t \in \lang$.  
\end{definition}

Consider an SRS $\rew$ consistent with $\lang$. 
Observe that due to transitivity and symmetry of equivalence, if words $s,t$ have the same normal form,
then either both belong to $\lang$ or none of them belongs to $\lang$. 
Furthermore, if $s \rewrites_{\rew} t$, then for every word $u$, we have 
$su \rewrites_{\rew} tu$. It follows that $s$ and $\normalForm{s}$ belong to the same equivalence class of the right 
congruence relation for $\lang$.

\begin{example}
\label{ex:reset-word}
Consider $\Sigma = \set{a,b}$, a word $w \in \Sigma^*$ and $\rewReset[w] = \{ aw \rightarrow w, bw \rightarrow w \}$.
Suppose that $\rewReset[w]$ is consistent with a regular language $\lang$. 
Then, for any word with the infix $w$, which can be written as $x w y$, it belongs to $\lang$ if and only
if $w y$ belongs to $\lang$, i.e., any prefix before $w$ does not influence whether 
the whole word belongs to $\lang$. 
We can show that for the minimal DFA $\aut$ for $\lang$, the word $w$ is a \emph{synchronizing word}, i.e., 
there is a state $s$ of $\aut$ such that for all states $q$ of $\aut$ we have
$\delta_{\aut}(q, w) = s$, all states lead to the same state $s$ upon reading $w$.
\end{example}

The active learning problem with SRS advice is to construct a minimal DFA recognizing 
a regular language $\lang$ having access to an oracle for $\lang$ and an SRS consistent with $\lang$.

\begin{definition}
The active automata learning with advice problem is as follows:
\begin{itemize}
    \item \textbf{Input}: an oracle $\oracle$ answering membership and equivalence queries about the target 
    regular language $\lang$, and an \emph{advice} SRS $\rew$ consistent with $\lang$.
    \item \textbf{Output}: a minimal-size DFA recognizing $\lang$.
\end{itemize}
\end{definition}

To reduce the number of queries, our algorithm, which contains  the \lStar-algorithm as a blackbox,
intercepts all queries asked by the \lStar-algorithm, and attempts to infer the answer using the advice SRS.
Only if inference fails, the query is forwarded to the oracle (Figure~\ref{fig:architecure}).
The advantage of this architecture is that it works with any implementation of the \lStar-algorithm 
(e.g. the observation table based, the TTT-algorithm~\cite{TTTalgorithm} or the L$^\#$-algorithm~\cite{Lsharp}).
We discuss how to process equivalence and membership queries separately. 

\subsection{Improving query complexity for DFA}
\label{s:query-improvment}

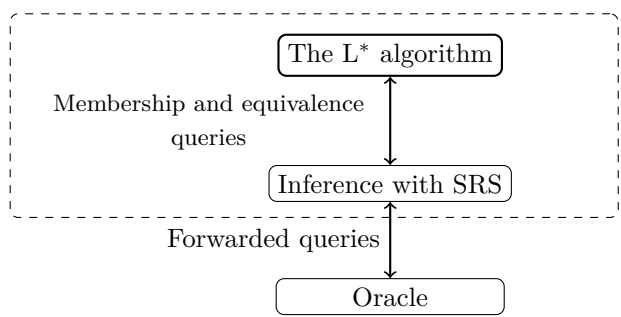
\begin{figure}[!ht]
\centering
\begin{tikzpicture}
\node [rectangle, rounded corners = 3.0,draw,thick ] (Lstar) at (0,0) {The L$^*$ algorithm};
\node [rectangle, rounded corners = 3.0,draw] (Inference) at (0,-1.7) {Inference with SRS};
\node [rectangle, rounded corners = 3.0,draw,dashed, minimum width=8cm, minimum height=2.7cm] (Algorithm) at (-1,-0.8) {};
\node [rectangle, rounded corners = 3.0,draw,minimum width =3cm] (Oracle) at (0,-3.2) {Oracle};

\draw[<->, thick] (Lstar) to node[left] {\begin{tabular}{c} Membership and equivalence \\ queries \end{tabular}}  (Inference);
\draw[<->, thick] (Inference)to node[left] {Forwarded queries}  (Oracle);
\end{tikzpicture}
\caption{An overview of the learning algorithm with advice}
\label{fig:architecure}
\end{figure}

To efficiently infer answers to membership queries, we require the SRS to be convergent, 
as inferring answers from advice involves computing normal forms.
To infer answers to equivalence queries, the SRS does not have to be convergent, but it has to be 
over the same alphabet as the target language (see Remark~\ref{rem:external-symbols}). 
Advice can be used for each type of query independently; e.g. having an non-convergent SRS, we can still
reduce the number of equivalence queries alone.

\subsubsection{Equivalence queries} 
The answer to the equivalence query can be positive only if a given candidate DFA $\aut$ is consistent with $\rew$.
If there are words $x,y$ witnessing inconsistency, i.e.,  $x \rewrites y$, but
$\neg(x \in \lang(\aut)  \Iff y \in \lang(\aut))$, then one of $x, y$ is a counterexample.
The algorithm can learn which of them is a counterexample with a single membership query.
Otherwise, if $\rew$ is consistent with the language of $\aut$, the algorithm forwards 
the equivalence query to the oracle.

The algorithm deciding consistency of $\rew$ with $\lang(\aut)$ 
is based on the following observation: 
\begin{lemma}
\label{l:compute-inconsistency-witnesses}
Let $\rew$ be an SRS over $\Sigma$ and $\aut$ be a minimal DFA over $\Sigma$ as well. 
The SRS $\rew$ is \emph{consistent} with $\lang(\aut)$ if and only if 
for every state $q$ of $\aut$ and every rule $l \rightarrow r \in \rew$ we have 
$\delta_{\aut}(q,l) = \delta_{\aut}(q,r)$.
\end{lemma}
\begin{proof}
For the implication from left to right, consider $x,y$ violating consistency, i.e., 
$x \rewrites y$, but $x \in \lang(\aut) \Iff y \in \lang(\aut)$ does not hold.
Consider a rewriting sequence $x = x_0 \rewritesOneStep x_1 \rewritesOneStep \ldots \rewritesOneStep x_k = y$. 
There is a pair $x_i, x_{i+1}$ such that $x_i \in \lang(\aut) \Iff x_{i+1} \in \lang(\aut)$ does not hold.
Thus, then there are $x',y'$ such that $x' \rewritesOneStep y'$ and $x' \in \lang(\aut) \Iff y' \in \lang(\aut)$ does not hold. 
Since $x'$ rewrites to $y'$ in one step, there is a rewrite rule $l \rightarrow r$ in $\rew$ and we have
$x' = u l v$, and $y' = u r v$. Let $q$ be a state of $\aut$ reached from the initial state upon $u$, and let 
$s_1, s_2$ be states reached from $q$ upon $l$ and $r$ respectively. 
Now, one of the states $\delta_{\aut}(s_1,v)$ and $\delta_{\aut}(s_2,v)$ is accepting and one is rejecting, 
and hence $s_1 \neq s_2$.

Conversely, assume that there is a rewrite rule $l \rightarrow r \in \rew$ and a state $q$ 
such that $s_1 = \delta(q,l) \neq \delta(q,r) = s_2$.
Let $u$ be a word upon which $\aut$ reaches $q$.
Since $\aut$ is minimal, Myhill-Nerode theorem implies there is a word $v$ distinguishing states $s_1$ and $s_2$, 
i.e., exactly one of the states: $\delta_{\aut}(s_1,v)$ and $\delta_{\aut}(s_2,v)$ is accepting. 
It follows that exactly one of the states: $\delta_{\aut}(q_0,ulv)$ and $\delta_{\aut}(q_0,urv)$ is accepting. 
Therefore, words $ulv$ and $urv$ satisfy $ulv \rewritesOneStep urv$ and $ulv \in \lang \Iff urv \in \lang$ does not hold.
\end{proof}

\Paragraph{Checking consistency}  
The algorithm verifying consistency is presented as Algorithm~\ref{a:consistency}.
If it finds $l \rightarrow r$ in $\rew$ and a state $q$ such that 
$\delta_{\aut}(q,l) \neq \delta_{\aut}(q,r)$, then it computes shortest words $u,v$ such that 
$\delta(q_0,u) = q$ and $v$ distinguishes $\delta(q,l)$ and $\delta(q,r)$.
Then, exactly one of the words $ulv$ and $urv$ is a counterexample, 
because  $ulv \in \lang \Iff urv \in \lang$ holds, while
$ulv \in \lang(\aut) \Iff urv \in \lang(\aut)$ does not hold. 

\Paragraph{Complexity} 
The algorithm checking consistency iterates over all rules from $\rew$ and all states in $\aut$.
In each iteration it has to compute $\delta(q,l)$ and $\delta(q,r)$ which amounts to 
$\bigO(\textrm{size}(\rew)\cdot|\aut|)$.
In the finial iteration it has to compute words $u, v$. 
The word $u$ can be found with a simple reachability check in $\bigO(|\aut|)$, and
finding $v$ requires dynamic programming to run in $\bigO(|\aut|)$ using the following lemma.

\begin{lemma}
Given a minimal DFA $\aut$ and its states $q, q_1, q_2$ such that $q_1 \neq q_2$,
\begin{itemize}
\item a word $u$ such that $\delta_{\aut}(q_0, u) = q$ can be computed in $\bigO(|\aut|)$,
\item a word $v$ distinguishing $q_1$ and $q_2$, i.e., exactly of state of $\delta_{\aut}(q_1, v), \delta_{\aut}(q_2, v)$ is accepting, can be computed in $\bigO(|\aut|)$.
\end{itemize}
\end{lemma}
\begin{proof}
A word $u$ can be found by BFS over $\aut$ considered as a labeled graph. The complexity of BFS is $\bigO(|V|+|E|)$, but it our case $|V|, |E|$ are bounded by $|\aut|$.

We define a sequence of relation $\sim_i$ on $Q_{\aut}$ as follows:
for $i =0$, we define $s_1 \sim_0 s_2$ if and only if $s_1, s_2 \in F$ or $s_1, s_2 \notin F$.
For $i \geq 0$, we define $s_1 \sim_{i+1} s_2$ if and only if 
for all $a \in \Sigma \cup \set{\epsilon}$ we have 
$\delta_{\aut}(s_1,a) \sim_{i} \delta_{\aut}(s_2,a)$.

All relations $\sim_i$ are equivalence relations on $Q$.
Observe that if for some $i$ we have $\sim_{i} = \sim_{i+1}$, then for all $k \geq i$ we have $\sim_{i} = \sim_{i+1}$.
In follows that for $n = |Q|-1$ we have $\sim_{n} = \sim_{n+1}$.

Observe that since $\aut$ is minimal, each equivalence class of $\sim_{n}$ is a singleton consisting of a single state. 
Finally, if $q_1$ and $q_2$ are not in the relation $\sim_{i}$, then there is a word of length at most $i$ distinguishing them.
This word can be computed as follows: if $q_1 \sim_{i} q_2$ does not hold, then either $i = 0$ and exactly on of the states is accepting, and hence $v = \epsilon$, or
$i >0$ and there is $a \in \Sigma \cup \set{\epsilon}$ such that $q_1 = \delta_{\aut}(s_1,a)$, $q_2 = \delta_{\aut}(s_2,a)$ and 
$s_1 \sim_{i-1} s_2$ does not hold.
\end{proof}

However, since 
$\aut$ has been produced by the \lStar-algorithm, 
words $u,v$ are in the internal data of the \lStar-algorithm; $u$ is the selector for $q$ and $v$ is among 
test words.

\Paragraph{Answering equivalence queries}
 In summary, before asking an equivalence query, the algorithm checks consistency of the SRS with
the candidate automaton $\aut$ (Algorithm~\ref{a:consistency}).
If it gets a pair of words $x,y$, it asks a membership query for $x$ to see whether $x$ distinguishes $\lang$ and $\lang(\aut)$ and hence it is a counterexample.
If not, then $y$ has to be a counterexample. 
Finally, if the SRS is consistent with $\aut$, the algorithm 
asks an actual equivalence query.

\begin{algorithm}
\caption{Check Consistency}{}
\label{a:consistency}
\begin{algorithmic}[1]
\renewcommand{\algorithmicrequire}{\textbf{Input:}}
\renewcommand{\algorithmicensure}{\textbf{Output:}}
\State{\textbf{input:} SRS $\rew$ and a minimal DFA $\aut$}
\For{$(l,r) \in \rew$}
   \For{$q \in Q_{\aut}$}
       \If{$\delta_{\aut}(q,l) \neq \delta_{\aut}(q,r)$}
           \State{Find a word $u$ such that $\delta_{\aut}(q_0,u) = q$}
           \State{Let $q_1 = \delta_{\aut}(q,l)$ and $q_2 = \delta_{\aut}(q,r)$}
           \State{Find a word $v$ distinguishing states $q_1$ and $q_2$, i.e., 
           exactly of one $\delta_{\aut}(q_1,v)$, $\delta_{\aut}(q_2,v)$ is accepting.}
            \State{\textbf{return} $(ulv, urv)$}
         
       \EndIf 
   \EndFor
\EndFor
\State{\textbf{return} Consistent}
\end{algorithmic}
\end{algorithm}

\begin{remark}[Additional symbols in the rewriting system]
\label{rem:external-symbols}
\newcommand{\rewCom}{\rew_{\textrm{com}}}
We require that the alphabets of the automaton and the SRS coincide, i.e., the rewriting system has no additional 
letters. It enables reduction of checking consistency to reasoning about single-step rewriting. 
Having additional letters in the SRS would extend the expressive power of SRS, but it would make
checking consistency harder. 
Suppose that an SRS is over $\Gamma \cup \Sigma$ and a DFA is over $\Sigma$ only and 
consider $x,y \in \Sigma^*$ such that $x \rewrites y$. 
It is possible that in a long chain of single step rewrites steps between 
$x$ and $y$ all intermediate words contain letters from $\Gamma$ and hence cannot be evaluated by the DFA.
Therefore, consistency of the single-step rewriting relation restricted to $\Sigma^*$ does not imply
consistency of the rewriting relation restricted to $\Sigma^*$. 
Furthermore, we cannot directly adapt our method based on computation of normal forms as the 
normal form of a word from $\Sigma^*$ need not belong to $\Sigma^*$.

To use normal forms, we can adapt the definition of consistency as follows:
$\rew$ is \emph{consistent} with $\lang$ if for all words $x,y \in \Sigma^*$, if 
$\normalFormSRS{\rew}{x} = \normalFormSRS{\rew}{y} \in (\Sigma \cup \Gamma)^*$,
 then $x \in \lang \Iff y \in \lang$.
Still, this does not solve the problem as checking consistency in this variant is undecidable.
Consider $\Gamma$ to be a primed copy of $\Sigma$ and suppose that some letters in $\Gamma$ only are commutative, 
i.e., an SRS $\rewCom$ consists of rules $ab \rightarrow ba$ for some pairs of letters in $\Gamma$ as well as
 rewrite rules allowing for rewriting each letter from $\Sigma$ to the corresponding copy in $\Gamma$. 
Then, one can easily reduce the intersection problem for regular trace languages, 
which is undecidable~\cite{bertoni1982equivalence}, to checking consistency of $\rewCom$ with a given language. 
\end{remark}

\subsubsection{Membership queries} 
To infer answers to membership queries, the SRS $\rew$ is required to be convergent.  
Our algorithm runs the \lStar-algorithm and keeps \emph{cache} of all membership queries in a dictionary, 
which stores normal forms of queried words along with the answers.
Then, when the \lStar-algorithm asks a membership query for $w$, the algorithm computes 
the normal form $\normalForm{w}$ and returns the answer from the cache if it
is present. 
Otherwise, the algorithm forwards the membership query to the oracle and upon receiving the answer, it saves
$\normalForm{w}$ with the answer in the cache and forwards the answer back to the \lStar-algorithm.

Consistency of $\rew$ with the target language $\lang$ implies that 
if  $\normalFormSRS{\rew}{s} = \normalFormSRS{\rew}{t}$, then $s \in \lang \Iff t \in \lang$.

The proposed approach has low overhead as it involves computing the normal form of a single word and a single
dictionary look-up.

\subsection{Advice for Mealy machines and other automata}

We discuss how to adapt the active learning with advice framework to other automata using 
Mealy machines as a leading example to motivate our approach (see~\cite{DBLP:conf/fm/ShahbazG09} for
details regarding active learning for Mealy machines.)

\Paragraph{Failed attempt} For DFA, consistency states that rewriting relation preserves the output of the automaton (acceptance/rejection).
For Mealy machines the output is a word
and hence the first idea would be to require that if $x \rewrites y$, then output words on $x$ and $y$ are equal.
That approach would only admit advice SRSs that do not change the length of word as the length of the output word 
is the same as the length of input.
In particular, Example~\ref{ex:reset-word} would not extend to Mealy machines.  
Instead, to adapt the consistency we employ the fundamental Myhill-Nerode theorem.

\Paragraph{Myhill-Nerode theorem for Mealy machines}
Let $f$ be a function from $\Sigma^*$ to $\Gamma^*$. 
We define $P^f \colon \Sigma^* \to \Gamma$ such that $P^f(w)$ is the last letter of $f(w)$. 
If $f$ is computed by a Mealy machine $\M$, then 
$P^f(w)$ is the last output letter returned by $\M$ processing $w$. 
The function $P^f$ is a crux for a counterpart of Myhill-Nerode theorem for Mealy machines. 
Let $\rightCongF$ be an equivalence relation on $\Sigma^*$ such that 
$u \rightCongF v$ if and only if for all words $w$ we have 
$P^f(uw) = P^f(vw)$.
The function $f$ is computed by a Mealy machine if and only if $\rightCongF$ has finitely
many equivalence classes~\cite{DBLP:phd/de/Niese2003}. Moreover, the states of the minimal Mealy machine for $f$ correspond to 
equivalence classes of $\rightCongF$.

\Paragraph{Advice for Mealy machines}
Based on the function $P^f$ we adapt the consistency to Mealy machines. 
The SRS $\rew$ is \emph{consistent} with $f \colon \Sigma^* \to \Gamma^*$ if for all words $s,t$, if
$s \rewrites t$, then $P^f(s) =  P^f(t)$.  
Then, we define the learning framework with SRS advice in the same way as for DFA.
Observe that inference of the answers to queries presented in Section~\ref{s:query-improvment} can
be straightforwardly adapted to Mealy machines.

\begin{example}
Recall the SRS $\rew_w$ from Example~\ref{ex:reset-word}.
Consider a function $f \colon \Sigma^* \to \Gamma^*$  computed by a Mealy machine such that $\rew_w$ is consistent with $f$.
Then, consistency implies that for any input words of the form $x w y$ and $x' w y$, 
the last $|y|$ in the output words $f(x w y)$ and $f(x' w y)$ are the same.
Furthermore, as in the case of DFA, we can show that for the minimal Mealy machine computing $f$ the word $w$ is a reset word. 
\end{example}

\Paragraph{Advice for other automata models}
We can employ these ideas to adapt advice to other automata models, which admit counterparts of Myhill-Nerode theorem
such as weighted automata~\cite{DBLP:books/ems/21/DrosteK21}, 
quantitative automata~\cite{DBLP:journals/tocl/ChatterjeeDH10} or tree automata~\cite{DBLP:books/ems/21/LodingT21}.

\section{Examples of advice}
\label{s:power}

\newcommand{\langAdd}{\lang_{\textrm{Add}}}
\newcommand{\autExt}{\widehat{\aut}}
In this section, we discuss the properties modeled by string rewriting systems.
The presented properties are ubiquitous among systems and easy to discover.
Therefore, SRS advice can often be effortlessly produced. 

\subsection{Idempotent actions}
Reset words presented in Example~\ref{ex:reset-word} are an example of an \emph{idempotent} action, i.e., 
an action which repeated multiple times has the same effect as a single action. 
Such actions are considered to be safe to repeat in case of 
(possible) failure. For this reason idempotent actions frequently appear in design of 
hardware~\cite{DBLP:conf/micro/KruijfS11} and software~\cite{DBLP:conf/popl/RamalingamV13}.
Examples of idempotent actions are writing to a register, closing a TCP connection~\cite{tcp-learning}, 
PUT, DELETE, request methods in HTTP (RFC 5741), and others. 
Observe that in these examples idempotence is either clear from the description or it is stated in the specification.
Therefore, idempotent actions can be easily identified.

We can model that a word $u$,  corresponding to a sequence of actions, 
is idempotent with the rewrite rule $ u u \rightarrow u$. 
In particular, $u$ can be a single letter.

\subsection{Concurrent and distributed systems}
Concurrent and distributed systems consist of several components running independently. 
Each of the components is modeled by an automaton (more precisely a Mealy machine) with two types of input letters:
\emph{input actions}, which correspond to communication coming to the component from the outside world, and
\emph{internal actions}, which correspond to internal state change of the component and are not visible
to the outside world~\cite{DBLP:conf/podc/LynchT87,DBLP:journals/corr/abs-2001-04235}.
Then, all automata corresponding to particular components are combined with \emph{parallel composition} into a single automaton,  
the \emph{system automaton}. The parallel composition is essentially the product of all the components.

Both input and internal actions are modeled with input letters in each component automaton, while
the input alphabet of the system automaton is the union of all input and internal actions of all components.
Thus, an input word corresponds to some interleaving of actions of different components.
As internal actions are not visible by other components, they are commutative, i.e., executions with swapped 
subsequent internal actions are semantically equivalent.

We can model insensitivity to the order of internal actions with an SRS.
Let $\Sigma$ be an alphabet of all actions and assume some linear order $\leq$ on it.
Consider a binary \emph{independence relation} $\Ind$ on $\Sigma$ defined as
$(a,b) \in \Ind$ if $a,b$ are internal actions from different components.
We define $\rew_{\Ind} = \set{ ba \rightarrow ab \mid (a,b) \in \Ind \land a < b}$.
The SRS $\rew_{\Ind}$ expresses that internal actions from different components can be swapped freely.
The linear order is to avoid non-terminating rewriting $ab \rightarrow ba \rightarrow ab \ldots$.

\subsection{Bit-wise addition and other invariants}
\label{s:bit-wise-addition}
\newcommand{\rewBin}{\rew_{\textrm{bin}}}
Consider the language encoding the relation of bit-wise addition.
Let $\Sigma = \set{0,1}^3 $.
Then, every word $w = (a_1, b_1, c_1) \ldots (a_k, b_k, c_k)$ over $\Sigma$ corresponds to three bit vectors $a_1 \ldots a_k$,
$b_1 \ldots b_k$, and $c_1 \ldots c_k$, which in turn correspond to three numbers $x_w, y_w, z_w$ represented 
in the least-significant-bit-first manner.
Consider the language $\langAdd = \set{w \mid x_w + y_w = z_w}$. 
The language $\langAdd$ is invariant over swapping $a_i$ with $b_i$, i.e., 
letters $(1,0,0), (1,0,1)$ can be rewritten respectively to $(0,1,0)$ and $(0,1,1)$, i.e., 
we consider an SRS $\rewBin = \{ (1,0,0) \rightarrow (0,1,0), (1,0,1) \rightarrow (0,1,1)\}$. 
Applying advice from $\rewBin$ effectively reduces the alphabet from $8$ to $6$ letters, and hence
it reduces the number of membership queries.

This is an example of a general scenario, where the specified function 
is invariant over some operations that can be expressed via word rewriting.

\section{Experiments}
\label{s:experiments}
We present the empirical evaluation of the active learning algorithm with advice.

\Paragraph{Implementations of the \lStar-algorithm} 
We have a custom implementation of the original \lStar-algorithm~\cite{angluin1987learning} as well
as an implementation on top of a state-of-the-art automata learning library \learnlib~\cite{DBLP:conf/cav/IsbernerHS15}, 
which is open source. 
In the latter case, we have only implemented inferring answers to equivalence queries.

\Paragraph{Experimental data}
We consider two sources of regular languages for experiments. 
First, we consider languages recognized by \emph{random DFA}. 
The number of states of a random DFA is given as a parameter. 
Next, for all $q \in Q$ and $a \in \Sigma$, the target state $\delta(q,a)$ is picked at random with the uniform distribution over $Q$.
The set of accepting states is picked at random as well; each state is accepting with probability $1/10$.
The random seeds have been fixed for the sake of reproducibility of the results.
As a second source of data, we consider \emph{pattern languages}, which for a word $u \in \Sigma^*$, are defined as 
$\lang_u = \set{ w \mid w \text{ contains } u}$.
This is an important class as micro-controllers essentially implement multipattern matching 
(if some event happens, execute an appropriate interrupt).
Furthermore, random DFA have certain properties with high probability, and hence
one can argue that interesting DFA are rare (not random).

\Paragraph{Idempotent actions}
In experiments, we consider DFA, in which a designated letter $a$ is idempotent, i.e., 
 for every state $q$ we have $\delta(q,a) = \delta(q,aa)$.
In our own implementation of the \lStar-algorithm, we generated 100 random DFA over a $4$ letter alphabet 
with the random size between 500 and 1000.
Then, for each random DFA, we modify it randomly to ensure idempotency of $a$.
The advice SRS is defined as $\rewIdm[a] = \set{ aa \rightarrow a}$.
For 16 (out of 100) cases, we have observed increase (1\% to 9\% and 4.5\% on average) 
in the number of membership queries.
However, in all but one case, there has been 15\% on average decrease in the number of equivalence queries.
The reason for the increase is that some counterexamples are more informative than others.
In the implementation based on \learnlib{}, we have considered 15 random DFA over a $4$ letter alphabet
with the random size between 500 and 1000. The reduction of equivalence queries ranges from 2\% to 55\% with
the mean 22\%.
The summary statistics are in Table~\ref{tb:three-experiments}, while
full source data are in~\cite{experimentsRepo}.

\Paragraph{Concurrent and distributed systems}
We have conducted experiments with two components. For simplicity, we work with DFA, for which \emph{parallel composition}
corresponds to the following notion of the convolution of languages.
Consider an alphabet $\Sigma = \Sigma_1 \cup \Sigma_2$, where $\Sigma_1$ and $\Sigma_2$ need not be disjoint.
A word $w$ over $\Sigma$ is a \emph{convolution} of words $w_1 \in \Sigma_1^*$ and $w_2 \in \Sigma_2^*$ if and only if
for $i=1,2$, the projection of $w$ on $\Sigma_i$ yields $w_i$. 
The language $\lang$ is the \emph{convolution} of languages $\lang_1$ and $\lang_2$ if it consists of all convolutions 
of all pairs of words $w_1,w_2$ with $w_1 \in \lang_1$ and $w_2 \in \lang_2$.
Observe that if languages  $\lang_1$ and $\lang_2$ are regular, then their convolution is regular as well. 

If the target language is the convolution of some regular languages over $\Sigma_1$ and $\Sigma_2$, then
the \emph{private} letters from $\Sigma_1 \setminus \Sigma_2$ and $\Sigma_2 \setminus \Sigma_1$ are commutative, 
which can be expressed with 
the SRS $\rewConv$ consisting of rules $b a \rightarrow ab$ for all $a \in \Sigma_1 \setminus \Sigma_2$ 
and $b \in \Sigma_2 \setminus \Sigma_1$.
Note that the order of letters in $\rewConv$ is fixed, and hence it is terminating and confluent.
If $\Sigma_1$ and $\Sigma_2$ are disjoint, the normal form of a word $w \in \Sigma^*$ is a word of the form
$w_1 w_2$, where $w_1 \in \Sigma_1^*$ and $w_2 \in \Sigma_2^*$, which is obtained by simple sorting of letters. 

The first example group of concurrent systems consist of convolutions of two pattern DFA
with disjoint alphabets with randomly generated patterns. 
We considered two types of pattern DFA: \texttt{OR} type pattern DFA (word $w$ is accepted if any of the patterns occur in a $w$) 
as well as \texttt{AND} type pattern DFA (word $w$ is accepted only if all patterns occur in $w$).
Each component of the convolution is constructed based on two randomly chosen patterns of length 10 over 4 letter alphabet. 
We observe significant reduction in the number of equivalence queries as
the number of equivalence queries decreases from numbers ranging from 60 to 300 
to 1 to 6, which is a 91\%--98\% reduction. 
This number of equivalence queries is proportional to the maximum of the number of equivalence queries
needed to learn each component DFA independently. This is possible, because counterexamples for the convolution of DFA
often contain counterexamples for both component DFA. 
\begin{table}
    \centering
    \begin{tabular}{|c||c|c||c||c|c|c|}
        \hline
        \multicolumn{3}{|c|}{\textbf{Target language}} & \multicolumn{2}{|c|}{no advice} 
        & \multicolumn{2}{|c|}{ with advice} \\
        \hline
        Conv & $\aut_1$ & $\aut_2$ &  MQ & EQ & MQ & EQ \\
        \hline
            208 & 16 & 13 & 226043  & 84  & 106100  & 4  \\\hline
            494 & 26 & 19 & 1013981 & 124 & 380094  & 4  \\\hline
            720 & 30 & 24 & 2274211 & 199 & 830729  & 4  \\\hline
            924 & 28 & 33 & 3668847 & 268 & 1306835 & 4 \\\hline
            952 & 28 & 34 & 4373192 & 220 & 1162988 & 4  \\\hline
            \hline
            255 & 15 & 17 & 400319 & 153  & 177084 & 6  \\\hline
            352 & 22 & 16 & 801477 & 165 & 258904 & 3  \\\hline
            675 & 27 & 25 & 988409 & 154 & 391815 & 1  \\\hline
            676 & 26 & 26 & 1543634 & 168 & 591684 & 2 \\\hline
            783 & 27 & 29 & 2451237 & 240 & 868529 & 2  \\\hline
        \end{tabular}
    \caption{Query complexity of learning convolutions of (a)~two random \textbf{pattern} DFA (first 5 rows), 
    and (b)~two random DFA (remaining 5 rows)}
    \label{tbl:conv1}
    \end{table}

The second example group includes convolutions of two randomly generated DFA over disjoint alphabets. 
Each component of these convolutions has between 15 and 30 states. 
Among these examples, we see an equally impressive reduction in the number of equivalence queries, averaging 98\%.
The example data for both cases are presented in Table~\ref{tbl:conv1} and 
the summary statistics are in Table~\ref{tb:three-experiments}.

Finally, the third example group includes convolutions of two randomly generated DFA over alphabets 
$\Sigma_1, \Sigma_2$ resp., where the alphabets are not disjoint.
We investigated influence of the size of the intersection $\Sigma_1 \cap \Sigma_2$ on the number of inferred 
equivalence queries in the implementation based on \learnlib. 
We observed that even if $4$ letters are shared by both automata while each automaton has $1$
private letter, the reduction of the number of equivalence queries ranges between 83\% to 92\% 
 (from 43 -- 75 equivalence queries without advice to 5 -- 10 equivalence queries).
The source data are in~\cite{experimentsRepo}.

\begin{table}
    \centering
    \begin{tabular}{|c|c|c|c|c|c|c|}
        \hline
         Advice & \multicolumn{3}{|c|}{Membership  queries} & \multicolumn{3}{|c|}{Equivalence queries}\\
        \hline
         &  Min & Max & Mean & Min & Max & Mean\\
        \hline
         $\rewConv[p]$ & 47 & 83 & 61 & 91 & 98 & 94 \\ 
        \hline
         $\rewConv[r]$ & 45 & 82 & 64 & 96 & 99 & 98 \\ 
\hline
        $\rewIdm[a]$  & -9 & 20 & 5 & 0 & 29 & 15 \\
        \hline
        \end{tabular}
    \caption{The percentage decrease in query complexity with advice comparing 
    to the \lStar-algorithm without advice. Considered SRSs are for convolution of pattern DFA $\rewConv[p]$, 
    convolution of random DFA $\rewConv[r]$, and the idempotent letter $\rewIdm[a]$}
    \label{tb:three-experiments}  
\end{table}

\Paragraph{Bit-wise addition}
The minimal DFA recognizing the relation of bit-wise addition (see Section~\ref{s:bit-wise-addition}) has $3$ states. 
Our implementation of the \lStar-algorithm starts with the set $C$ that includes all the letters from the alphabet.
Therefore, the first candidate DFA is correct and hence the \lStar-algorithm 
asks $1$ equivalence query and $201$ membership queries to construct that automaton. 
The \lStar-algorithm with the SRS $\rewBin$ as advice asks only $115$ membership queries. 
Therefore, we have observed $42\%$ reduction in membership queries.

\section{Extensions}
\label{s:extensions}
We present two extensions of the framework presented in Section~\ref{s:framework}: 
an extension of SRS to a more general model of \emph{controlled string rewriting systems} and
a more flexible definition of consistency of an SRS with a language. 
Both extensions increase the expression power of the advice framework as we demonstrate with examples.

\subsection{Controlled string rewriting systems}

Controlled string rewriting systems are an extension of SRS, in which each rewriting rule comes with constraints
on positions at which it can be applied. E.g. a rule may be applied only at the first position in any word.

\Paragraph{Regular expressions} A regular expression $e$ over $\Sigma$ 
is an expressions consisting of 
the empty set $\emptyset$, the empty word $\epsilon$, 
letters $\Sigma$, concatenation $\cdot$, alternation $+$, and the Kleene star $^*$. 
The language defined by $e$, denoted by $\lang(e)$, is defined recursively as follows:,
$\lang(\emptyset) = \emptyset$,
$\lang(\epsilon) = \set{\epsilon}$,  
$\lang(a) = \set{a}$, for $a \in \Sigma$,
$\lang(e_1 \cdot e_2) = \set{wv \mid w \in \lang(e_1), v \in \lang(e_2)}$,
$\lang(e_1 + e_2) = \lang(e_1) \cup \lang(e_2)$, and
$\lang(e^*) = \bigcup_{i \geq 0} \lang(e^i)$, where $e^0 = \epsilon$ and $e^{i+1} = e \cdot e^i$.
The expression $\Sigma$ is considered as a shorted form of $a_1 + \ldots + a_n$, where 
${a_1, \ldots, a_n}$ are all letters of the alphabet.

\Paragraph{Controlled string rewriting systems} 
A \emph{controlled string rewriting system (cSRS)}
 $\rew$ over an alphabet $\Sigma$ is 
a finite set of quadruples $(l,r, e_x, e_y)$, where $l,r$ are words over $\Sigma$ and
$e_x, e_y$ are regular expressions over $\Sigma$.
For a cSRS $\rew$, we define a \emph{single-step rewrite relation} $\rewritesOneStep_{\rew}$ over 
words from $\Sigma^*$ as follows:
for all $s,t \in \Sigma^*$, we have $s \rewritesOneStep_{\rew} t$ if and only if 
there are words $x,y \in \Sigma^*$ and $(l,r,e_x,e_y) \in \rew$ such that $s = x l y$, $t = x r y$ and $x \in 
\lang(e_x), y \in \lang(e_y)$.
As for SRS, the relation $\rewrites_{\rew}$ is the transitive and reflexive closure of $\rewritesOneStep_{\rew}$. 
Observe that every SRS can be also considered as a cSRS as a rule $l \rightarrow r$ 
is semantically equivalent to $(l,r,\Sigma^*,\Sigma^*)$.

\Paragraph{cSRS advice} 
Controlled rewriting has been considered in the literature~\cite{DBLP:conf/icalp/Chottin79,DBLP:conf/frocos/JacquemardKS11,DBLP:journals/jsyml/ChvalovskyH16}.
Confluence, termination, and consistency with a regular language defined for SRS extend seamlessly to cSRS, 
and hence we can consider cSRS advice in active automata learning. Furthermore, \emph{cache} for
membership queries can be implemented in the same way for cSRS as for SRS. We discuss how to adapt
checking consistency of an SRS with $\lang(\aut)$ to cSRS (Lemma~\ref{l:compute-inconsistency-witnesses}).

\Paragraph{Equivalence queries with cSRS} 
Consider a cSRS $\rew$ over $\Sigma$ and  a minimal DFA $\aut$ over $\Sigma$ as well.
As in the poof of Lemma~\ref{l:compute-inconsistency-witnesses}, we observe that 
consistency of single-step rewriting  $\rew$ with the language $\lang(\aut)$ implies 
 consistency of $\rew$ with $\lang(\aut)$.
Now, consider $(l,r,e_x,e_y) \in \rew$ and words $s,t$ such that $s$ rewrites in one step to $t$
 and $s \in \lang(\aut) \Iff t \in \lang(\aut)$ does not hold.
Then, for some words $x,y$ we have $s = x l y$, $t = x r y$, and
$x \in \lang(e_x)$, $y \in \lang(e_y)$. 
Therefore, for every rule $(l,r,e_x,e_y)$, we need to check whether for every state $q$ reachable
with some word from $\lang(e_x)$, states $\delta_{\aut}(q,l)$ and $\delta_{\aut}(q,r)$ are
indistinguishable with any word from $\lang(e_y)$. These conditions can be verified in polynomial time in 
$|\aut|, |e_x|$ and $|e_y|$. In consequence, we have:
\begin{lemma}
Consistency of a given cSRS $\rew$ over $\Sigma$ with the language $\lang(\aut)$ of a given DFA $\aut$ over $\Sigma$ 
can be decided in polynomial time in $|\rew| + |\aut|$.  
\end{lemma}

\begin{proof}
For the implication from left to right, consider $w,v$ violating consistency, i.e., 
$w \rewrites v$, but $w \in \lang(\aut) \Iff v \in \lang(\aut)$ does not hold.
Consider a rewriting sequence $w = w_0 \rewritesOneStep w_1 \rewritesOneStep \ldots \rewritesOneStep w_k = v$. 
There is a pair $w_i, w_{i+1}$ such that $w_i \in \lang(\aut) \Iff w_{i+1} \in \lang(\aut)$ does not hold.
Thus, then there are $w',v'$ such that $w' \rewritesOneStep v'$ and $w' \in \lang(\aut) \Iff v' \in \lang(\aut)$ does not hold. 
Since $w'$ rewrites to $v'$ in one step, there is a rewrite rule $(l,r,e_x,e_y)$ in $\rew$ and we have
$w' = x l y$, and $v' = x r y$ with $x \in \lang(e_x)$ and $y \in \lang(e_y)$.
Let $q$ be a state of $\aut$ reached from the initial state upon $x$.
Clearly, $q$ is reachable over $e_x$.
Now, let $s_1, s_2$ be states reached from $q$ upon $l$ and $r$ respectively. 
Now, one of the states $\delta_{\aut}(s_1,y)$ and $\delta_{\aut}(s_2,y)$ is accepting and one is rejecting, 
and hence $s_1, s_2$ are distinguished over $e_y$.

Conversely, assume that there is a rewrite rule $(l,r,e_x,e_y) \in \rew$ and a state $q$ reachable over $e_x$. 
such that $s_1 = \delta(q,l) \neq \delta(q,r) = s_2$ are distinguished over $e_y$.
Let $x \in \lang(e_x)$ be a word upon which $\aut$ reaches $q$ and let
$y \in \lang(e_y)$ be a word distinguishing $s_1$ and $s_2$.
It follows that exactly one of the states: $\delta_{\aut}(q_0,xly)$ and $\delta_{\aut}(q_0,xry)$ is accepting. 
Therefore, words $xly$ and $xry$ satisfy $xly \rewritesOneStep xry$ and $xly \in \lang \Iff xry \in \lang$ does not hold.
\end{proof}

The reduction in query complexity depends on advice  and it can range from none to almost all. 
In particular, as we will see, a cSRS can encode the DFA recognizing the specified language.
While it is not a reasonable scenario to specify the complete DFA using an SRS 
as it makes learning unnecessary, 
providing a partial specification with an SRS is a sensible approach.
For instance, it is often clear which sequences of actions lead to an error state.
We show how to encode \emph{partial DFA}, which are a variant of DFA with undefined transitions 
(the transition function is a partial function);  
a partial DFA can be regarded as a partial specification. 

Consider a partial DFA $\autB$ over the alphabet $\Sigma$. 
We define an SRS $\rew[\autB]$ over $\Sigma$ such that for 
all words $u,v \in \Sigma^*$, if runs over $u$ and $v$ in $\autB$ are defined 
(all transitions along words $u$ and $v$ are defined) and terminate in the same state, then
$\normalForm{u} = \normalForm{v}$.
Let $S$ be a prefix-closed set of \emph{access words} for all states of $\autB$, where each state has exactly 
one corresponding word. 
Then, $\rew[\autB]$ consists of rules $(u a, u',\set{\epsilon},\Sigma^*)$, for all words $u \in S$ and $a \in \Sigma$ such that $ua \notin S$
and $u' \in S$  satisfies $\delta_{\autB}(q_0,u') = \delta_{\autB}(q_0,ua)$.
The rule $(u a, u',\set{\epsilon},\Sigma^*)$ states that $ua$ can be rewritten to $u'$, but 
only at the beginning of the word, i.e., $uay$ rewrites to $u'y$.
Clearly, for every word $w$, if the run over $w$ is defined, then 
the word $w$ rewrites to the unique $u \in S$ such that 
$\delta_{\autB}(q_0,w) = \delta_{\autB}(q_0,u)$. 

The number of rules SRS in $\rew[\autB]$ is the number of transitions of $\autB$.
One could drop the restriction of the application position of the rule $ua \rightarrow u'$
and allow rewriting at any position, but then the SRS would correspond to a congruence relation w.r.t. 
$\lang(\autB)$, which may have 
exponentially many equivalence classes w.r.t. $|Q|$ and hence the SRS could have exponential size as well.

We consider 100 random DFA with the size between 500 and 1000. 
For each random DFA $\aut$, we construct a partial DFA $\autB$ from $\aut$ by random pruning
all but 10--20 randomly chosen transitions.
Based on $\autB$, we construct the SRS $\rew[{\autB}]$,
which has between 10 and 20 rules.
We have observed $26\%$ decrease in the number of membership queries and small decrease in the number of
equivalence queries~\cite{experimentsRepo}.

\subsection{One-sided advice}
\label{s:one-sided}
The advice mechanism discussed until now is based on the equivalence $s \in \lang \Iff t \in \lang$.
In this section we consider a relaxed notion of advice, \emph{one-sided advice}, which is based on the implication.
It is easier to fit fine-grained one-sided advice consistent with a (target) language than two-sided advice.
For instance, it is often clear which sequence of actions are invalid and hence lead to an error state.
On the negative side, as we discuss below, it is computationally harder to infer information from one-sided advice.

\Paragraph{Positive and negative consistency}
To define one-sided advice, the consistency notion splits to \emph{positive} and \emph{negative consistency}.
For an SRS $\rew$ over $\Sigma$ and a language $\lang \subseteq \Sigma^*$, we say that
(a)~$\rew$ is \emph{positively consistent} with $\lang$ if for all words $s,t$, 
if $s \rewrites_{\rew} t$ , then $s \in \lang \Rightarrow t \in \lang$, and
(b)~$\rew$ is \emph{negatively consistent} with $\lang$ if for all words $s,t$, 
    if $s \rewrites_{\rew} t$, then $s \notin \lang \Rightarrow t \notin \lang$.  
Note that an SRS $\rew$ is consistent with $\lang$ if and only if it is
both positively and negatively consistent with $\lang$.

\begin{example}
We can specify with a single membership query that a language $\lang$ contains all
words with the infix $u$. 
Consider a word $u$ and an SRS $\rew_u = \set{ u \rightarrow a u, u \rightarrow u a \mid a \in \Sigma}$. 
If $\rew_u$ is positively consistent with $\lang$ and $u \in \lang$, then all words with an infix $u$ belong to $\lang$.
Similarly, if $u$ does not belong to a language $\lang$ and $\rew_u$ is negatively consistent with 
$\lang$, then no word with the infix $u$ belongs to $\lang$. 
In that case, we need a single membership query as well.
\end{example}

Now, we discuss how to implement inference from one-sided advice for equivalence and  membership queries.
We discuss only the case of positive consistency as the other case is symmetric.

\subsubsection{Equivalence queries}
Consider an SRS $\rew$ and a language $\lang$ such that $\rew$ is positively consistent with $\lang$.
Inference for equivalence queries is similar to the case of two-sided advice. It is slightly more expensive,
but still it is polynomial in $|\rew|$ and $|\aut|$. 

Let $\lang_{s}(\aut)$ be the language of words accepted by $\aut$ starting from the state $s$.
We define $\leq_{\aut}$ on states of $\aut$ as $s_1 \leq_{\aut} s_2$ if and only if 
$\lang_{s_1}(\aut) \subseteq \lang_{s_2}(\aut)$. 

\begin{lemma}
\label{l:compute-inconsistency-witnesses-one-sided}
Let  $\rew$ be an SRS over $\Sigma$ and $\aut$ be a DFA over $\Sigma$. 
The SRS $\rew$ is positively consistent (resp., negatively consistent) with $\lang$ if and only if 
for all rewrite rules $l \rightarrow r \in \rew$ and every state $q$ of $\aut$ we have
$\delta_{\aut}(q,l) \leq_{\aut} \delta_{\aut}(q,r)$ 
(resp., $\delta_{\aut}(q,r) \leq_{\aut} \delta_{\aut}(q,l)$).
\end{lemma}

The above lemma is similar to Lemma~\ref{l:compute-inconsistency-witnesses}, but
the condition $\delta_{\aut}(q,l) = \delta_{\aut}(q,r)$
 is replaced with $\delta_{\aut}(q,l) \leq_{\aut} \delta_{\aut}(q,r)$ for positive consistency and
 $\delta_{\aut}(q,r) \leq_{\aut} \delta_{\aut}(q,l)$ for negative consistency.
We can compute in $O(|\aut|^3)$ the set of all pairs in the relation $\leq_{\aut}$.
This algorithm can be further optimized  with the information from the \lStar-algorithm. We discuss that in the following remark.

\begin{remark}[Computing the relation $s_1 \leq_{\aut} s_2$  with the information from the \lStar-algoroithm]
We sketch the algorithm computing the subsumption relation. It starts with the set of all pair of states 
$X_0 = Q \times Q$. Then in each iteration $i$, for all pairs $(s_1, s_2) \in X_i$ it checks 
whether for all $a \in \Sigma$, states $(s_1, s_2)$ moved by $a$ are in $X_i$, i.e.,
$(\delta(s_1,a), \delta(s_2,a)) \in X_i$. The set $X_{i+1} \subseteq X_i$ consists of all pair	s that satisfy 
this condition. The algorithm stops when $X_i = X_{i+1}$. 

Observe that this algorithm can be bootstrapped with the information from the \lStar-algorithm.
The input automaton $\aut$ is the one constructed by the \lStar-algorithm, that is, it is $\aut_{S,C}$ 
for some sets of words $S,C$.
Therefore, we can start with $X_0$ being the set of pairs $(s_1, s_2)$ such that 
for all $c \in C$ we have $s_1 c \in \lang \implies s_2 c \in \lang$. 
In our implementation of the \lStar-algorithm (and various other implementations), for all states $s \in S$ and 
all $c \in C$, the information whether $s c \in \lang$ is already computed for the construction of $\aut_{S,C}$.
Note, however, the condition: for all words $u \in C$ we have 
 $u \in \lang_{s_1}(\aut) \implies u \in \lang_{s_2}(\aut)$, does not imply that 
 $\lang_{s_1}(\aut) \subseteq \lang_{s_2}(\aut)$.
To see that cosider $\lang = \lang((a^3)^* + b(a^2 + a^2b))$.
Let $S = \{\epsilon, a, aa, b, ba, baa, baab, baaba\}$ and 
$C = \{ \epsilon, a, aa, aaa, b, ab, aab, aaab\}$.
Observe that $\aut_{S,C}$ recognizes $\lang$. However,
for the state corresponding to $a$, the only word $c \in C$ such that
$a c \in \lang$ is $aa$, while for the state that corresponds to $b$ there are two such words
$aa$ and $aab$ ($baa, baab \in \lang$). 
Yet, $\lang_a(\aut_{S,C}) = \lang(aa (a^3)^*)$ and $\lang_b(\aut_{S,C}) = \set{ aa, aab}$.
\end{remark}

As in Lemma~\ref{l:compute-inconsistency-witnesses}, having 
 a counterexample $q,s_1 = \delta_{\aut}(q,l),s_2 = \delta_{\aut}(q,r)$ to positive (resp,. negative) consistency, we can construct two 
 words for which one is the counterexample to the equivalence query. The construction is essentially the same:
 let $x,y$ be words such that $\delta(q_0, x) = q$, and $y$ be a word 
 in $\lang_{s_1}(\aut) \setminus \lang_{s_2}(\aut)$.
 Then, one of the words $xly, xry$ is the counterexample to the equivalence of $\lang(\aut)$ and $\lang$.

\subsubsection{Membership queries}
The answer to the membership query $x \in \lang$ can be inferred in two cases:
(a)~if the algorithm knows $y \in \lang$ such that $y \rewrites x$, then hence $x \in \lang$, and
(b)~if the algorithm knows $y \notin \lang$ such that $x \rewrites y$, then $x \notin \lang$. 
However, knowing whether $y \in \lang$ holds and $\normalForm{x} = \normalForm{y}$
is insufficient to deduce whether $x$ belongs to $\lang$. 
This makes inference computationally expensive for two reasons.

First, computing the normal from is typically easier than deciding $x \rewrites y$.
In the former, for a convergent $\rew$ it suffices to apply rewrite rules arbitrarily, since 
confluence guarantees the uniqueness of the resulting word. 
In contrast, to decide $x \rewrites y$ the choice of reductions is important.
There are examples of convergent SRSs, in which the normal form can be computed easily, 
but deciding $x \rewrites y$ is difficult (e.g. $\PSPACE$-hard).

Second, even if deciding $x \rewrites y$ is not difficult, the algorithm would need to test the queried word $x$ 
against all words in the cache, which is considerably more expensive than a single dictionary look-up.
Still, for specific SRSs we can implement dedicated efficient procedures.

\begin{example}
\newcommand{\rewUp}{\rew_{\uparrow}}
A language $\lang$ is \emph{upward-closed}, if for every word $x \in \lang$, for all words $y$ that contain $x$ 
as a subsequence belong to $\lang$~\cite{DBLP:journals/tcs/KarandikarNS16}.
We can express that the learned language $\lang$ is upward-closed with the following SRS 
$\rewUp = \{\epsilon \rightarrow a \mid a \in \Sigma \}$. 

Observe that for $\rewUp$ we can implement inference for membership queries efficiently. 
Basically, all words that are known to be in the learned language are stored in the suffix tree, and 
before asking the membership query to the oracle, the algorithm first checks using the (generalized) suffix tree, 
whether a given word has a subsequence in a set of words. In the worst case, the algorithm needs 
to traverse the whole suffix tree, which is still better than checking all words separately.
\end{example}

\section{Discussion and future work}
\label{s:future}
We have proposed an extended automata learning framework, which combines prior partial specification given via 
a (controlled) string rewriting system and learning from queries. 
We have discussed that partial specifications are often evident for the language under learning and 
hence producing them does not take much effort.
The proposed method achieves promising results on synthetic data, which are motivated by real-life scenarios.
We plan on applying the proposed framework to model learning and synthesis of circuits.

\bibliography{papers}

\end{document}